\documentclass[11pt,a4paper]{article}

\usepackage[utf8]{inputenc}
\usepackage[T1]{fontenc}
\usepackage[english]{babel}
\usepackage{lmodern}
\usepackage[a4paper,margin=2.5cm]{geometry}
\usepackage{microtype}
\usepackage{amsmath,amsthm,amssymb,mathtools}
\usepackage{booktabs}
\usepackage{enumitem}
\usepackage{graphicx}
\usepackage{subcaption}

\graphicspath{{images/}}

\usepackage{xcolor}
\usepackage{hyperref}
\usepackage{tikz}
\usetikzlibrary{arrows.meta,positioning,calc,fit,backgrounds,decorations.pathreplacing,shapes.geometric}
\usepackage{authblk}

\hypersetup{
  colorlinks=true,
  linkcolor=blue!40!black,
  citecolor=blue!40!black,
  urlcolor=blue!50!black,
  pdftitle={One Hex reduction to rule them all: Quoridor, Maze Attack, Pinko Pallino and Blockade are PSPACE-complete},
  pdfauthor={Francesco Carboni, Daniele Muscillo}
}

\setlist[itemize]{topsep=3pt,itemsep=2pt,parsep=0pt}
\setlist[enumerate]{topsep=3pt,itemsep=2pt,parsep=0pt}

\newtheorem{theorem}{Theorem}[section]
\newtheorem{lemma}[theorem]{Lemma}

\newtheorem{corollary}[theorem]{Corollary}
\theoremstyle{definition}
\newtheorem{definition}[theorem]{Definition}

\newtheorem{remark}[theorem]{Remark}

\newtheorem{fact}[theorem]{Fact}

\newcommand{\QWIN}{\textsc{Win(Quoridor)}}
\newcommand{\GHex}{$G(\text{graph-Hex})$}
\newcommand{\Open}{\perp}
\newcommand{\gwidth}{\operatorname{width}}
\newcommand{\gheight}{\operatorname{height}}

\title{One Hex reduction to rule them all:\\
Quoridor, Maze Attack, Pinko Pallino and Blockade\\ are PSPACE-complete}

\author[1]{Francesco~Carboni\thanks{\texttt{carboni.2058986@studenti.uniroma1.it}}}
\author[1,2]{Daniele~Muscillo\thanks{\texttt{muscillo.2080466@studenti.uniroma1.it}}}

\affil[1]{Sapienza University of Rome}
\affil[2]{Scuola Ortogonale, Elicsir Foundation}

\date{June 17, 2026}

\begin{document}
\maketitle

\begin{abstract}
\textsc{Quoridor} is a popular award-winning board game whose computational complexity, listed among the open problems of the Demaine--Hearn survey \cite{demaine-hearn}, remained open for nearly two decades. It was settled only recently \cite{drop2026quoridor}, via a reduction from the formula game $G_{pos}$ tailored to \textsc{Quoridor}. We give a shorter and more general proof: a single reduction from Reisch's planar graph-Hex \cite{reisch1981}, in which wall placement encodes the path-connection structure of Hex. The same construction settles three closely related games---\textsc{Maze Attack} and \textsc{Pinko Pallino} with no change, and \textsc{Blockade} with only minor adaptations---showing that all four are PSPACE-complete, the latter three for the first time.
More generally, our reduction shows that any \emph{race-and-wall} game is PSPACE-complete.
\end{abstract}

\section{Introduction}

The systematic study of the computational complexity of games has been one of the most fruitful meeting points between recreational mathematics and theoretical computer science. Ever since Schaefer showed that a handful of innocent-looking two-player games are as hard as any problem solvable in polynomial space \cite{schaefer1978}, a long line of work has placed classical board games at the top of the polynomial hierarchy: generalized \emph{Hex} and the \emph{Shannon switching game} \cite{eventarjan1976,reisch1981}, \emph{Gomoku} \cite{reisch1980gobang}, \emph{Go} \cite{lichtensteinsipser1980}, \emph{Amazons} \cite{furtak2005amazons}, and connection games in general \cite{bonnet2016connection} are all PSPACE-complete on boards of arbitrary size. Games whose play is not bounded in advance can be even harder: generalized \emph{chess} and \emph{checkers} are EXPTIME-complete \cite{fraenkel1981chess,robson1984checkers}. This boundary is captured cleanly by Hearn and Demaine's constraint-logic framework \cite{hearndemaine2009}, in which bounded two-player games are PSPACE-complete while unbounded ones are EXPTIME-complete. The race-and-wall games we study fall on the bounded side: once a wall is placed, it cannot be moved or removed, and hence the number of distinct wall configurations is finite. This is the key reason why the corresponding decision problems lie in PSPACE. The survey of Demaine and Hearn \cite{demaine-hearn} collects many of these results and, tellingly, also records the games whose complexity was still open---\textsc{Quoridor} among them.

\textsc{Quoridor} is among the most popular abstract strategy games of the last three decades. Designed by Mirko Marchesi and published by Gigamic in 1997, it won the Mensa Select award that same year and has since become a widely sold, award-winning title, with a particular following in the United States. Despite this visibility, its computational complexity---listed as open in \cite{demaine-hearn}---remained unresolved for nearly two decades.

\textsc{Quoridor} does not stand alone but belongs to a small family of two-player, perfect-information \emph{race-and-wall} games, in which players move pawns across a grid while dropping walls to obstruct the opponent. The oldest member is \textsc{Blockade}, marketed in Europe as \emph{Cul-de-sac}, created by Philip Slater and published by Lakeside in 1975 \cite{blockade1975}. Much of its mechanics resurface in \textsc{Pinko Pallino} and, later, in \textsc{Quoridor} itself, both attributed to Marchesi; \textsc{Maze Attack} is a more recent jump-less variant. The four games thus share a common ancestry and a common core mechanic, but differ in movement rules, wall types, and the move/place turn structure---differences that, as we show, do not affect their worst-case complexity.

\paragraph{Related work.}
The complexity of \textsc{Quoridor} was settled only very recently. Concurrently with and independently of the present work, Drop, Rin, and van der Velde~\cite{drop2026quoridor}---whose preprint appeared shortly before ours---proved \textsc{Quoridor} PSPACE-complete by a reduction from the formula game $G_{pos}$. Their argument, however, relies on \textsc{Quoridor}-specific pawn-interaction primitives, such as \emph{railroading}, which depend on the jump rule and on purely orthogonal movement. For this reason, it does not directly transfer to the related games considered here. Our approach is instead closer in spirit to the connection-game tradition, since it reduces from Reisch's planar graph-Hex~\cite{reisch1981} and uses only wall gadgets, making the construction robust under the rule variations of \textsc{Maze Attack}, \textsc{Pinko Pallino}, and \textsc{Blockade}.
More broadly, the complexity of concrete board games has continued to attract attention, with recent results settling \emph{Arimaa} \cite{rin2024arimaa} and \emph{Hive} \cite{andel2026hive}. Particularly close in spirit is the escape game \emph{Nemesis} of Bergé, Dailly, and Gerard \cite{berge2026nemesis}, in which a fugitive tries to reach an exit while an adversary deletes one edge per round, and which is again shown to be PSPACE-complete. Such pursuit-and-evasion games on graphs are conceptually adjacent to the race-and-wall games studied here, where one player races to a target side while the other erects walls to seal off every route.

\paragraph{Our contribution.}
We give a \emph{single} reduction from planar graph-Hex that settles four games at once. The construction never depends on the specific mechanics of any one game---jumps, diagonal moves, double pawns, or directional walls---so it applies verbatim to \textsc{Quoridor}, \textsc{Maze Attack}, and \textsc{Pinko Pallino}, and extends to \textsc{Blockade} with only the minor adaptations its second pawn and directional walls require. This gives a new, shorter, and arguably more elegant proof for \textsc{Quoridor} and the first proof for the other three.

The rest of the paper is organized as follows. Section~\ref{sec:games} describes the four games, their shared rules, and the parametric generalization on which the decision problems are based; Section~\ref{sec:pgh} recalls the planar graph-Hex problem we reduce from; Section~\ref{sec:embedding} presents the gadgets; and Section~\ref{sec:reduction} gives the reduction itself.

\section{The games and their generalization}
\label{sec:games}

\subsection{The four games}
\label{sec:fourgames}

\paragraph{The common framework.}
The four games we study are two-player, perfect-information, deterministic \emph{race-and-wall} games: each player controls one or more pawns that start on one side of a rectangular grid, and the goal is to be the first to bring a pawn to the opposite side. The players move in alternating turns. Besides moving pawns, a player may drop \emph{walls} on the grid lines to obstruct the opponent: a wall has costant length (two for these four games), blocks pawn movement across it, and may neither cross nor overlap another wall. In all four games, once a wall has been placed it is \emph{permanent}: it can never be removed nor moved for the rest of the game. This monotonicity is fundamental to our reduction, since it turns every wall placement into an irreversible commitment, mirroring the irreversible colouring of a vertex in Hex.\\
A second legality constraint is likewise shared by all four games: a wall may be placed only if it does not remove the last remaining path of any pawn to its target side. Walls can therefore lengthen a route arbitrarily but can never seal it off completely. However, we will consider this rule not to be assumed in general in \emph{race-and-wall} games.  The original games are played on fixed boards (\textsc{Quoridor} on $9\times 9$, Pinko Pallino on $11\times 11$, Blockade on $11\times 14$); for the complexity analysis these sizes are generalized to arbitrary $n$, as detailed in the next subsection.

\paragraph{Axes of variation.}
Within this common framework, the four games differ only along a few independent axes:
\begin{itemize}
\item \emph{Pawns per player:} one (\textsc{Quoridor}, \textsc{Maze Attack}, \textsc{Pinko Pallino}) or two (\textsc{Blockade}).
\item \emph{Pawn movement:} a single orthogonal step (\textsc{Quoridor}, \textsc{Maze Attack}); a single orthogonal or diagonal step (\textsc{Pinko Pallino}); or two orthogonal squares or one diagonal square (\textsc{Blockade}). See Figure~\ref{fig:moves}.
\item \emph{Turn structure:} either move a pawn or place a wall (\textsc{Quoridor}, \textsc{Maze Attack}, \textsc{Pinko Pallino}), versus move a pawn and then place a wall when one is available (\textsc{Blockade}).
\item \emph{Pawn confrontation:} when two pawns meet, \textsc{Quoridor} allows a jump over the opponent, with a lateral deviation if the landing square is blocked; \textsc{Maze Attack} forbids jumps, so the pawns block one another; \textsc{Blockade} again allows jumps. We were unable to find the confrontation rule of \textsc{Pinko Pallino} in the available sources; as we will see, however, our reduction works for any choice of confrontation rule. See Figure~\ref{fig:jumprules}.
\item \emph{Wall types:} a universal wall type that can be placed either horizontally or vertically (\textsc{Quoridor}, \textsc{Maze Attack}, \textsc{Pinko Pallino}), versus two distinct wall types, one restricted to vertical and one to horizontal placement (\textsc{Blockade}).
\end{itemize}

\begin{figure}[ht]
    \centering
    \begin{subfigure}[t]{0.45\textwidth}
        \centering
        \includegraphics[width=0.48\linewidth]{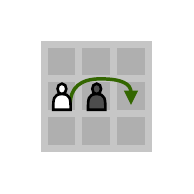}
        \hfill
        \includegraphics[width=0.48\linewidth]{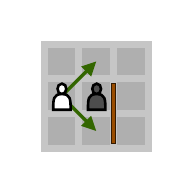}
        \caption{\textsc{Quoridor}}
        \label{fig:conf-quoridor}
    \end{subfigure}
    \hfill
    \begin{subfigure}[t]{0.22\textwidth}
        \centering
        \includegraphics[width=\linewidth]{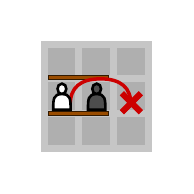}
        \caption{\textsc{Maze Attack}}
        \label{fig:conf-maze}
    \end{subfigure}
    \hfill
    \begin{subfigure}[t]{0.22\textwidth}
        \centering
        \includegraphics[width=\linewidth]{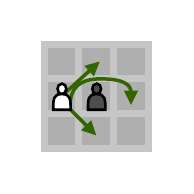}
        \caption{\textsc{Blockade}}
        \label{fig:conf-blockade}
    \end{subfigure}
    \caption{Pawn-confrontation rules across the games.}
    \label{fig:jumprules}
\end{figure}

\begin{figure}[h]
    \centering
    \begin{subfigure}[t]{0.22\textwidth}
        \centering
        \includegraphics[width=\linewidth]{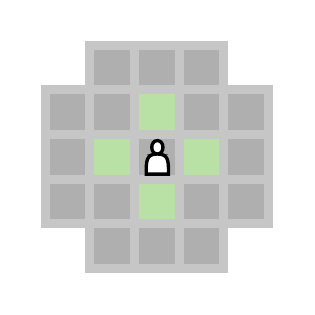}
        \caption{\textsc{Maze Attack} and \textsc{Quoridor}.}
    \end{subfigure}
    \hfill
    \begin{subfigure}[t]{0.22\textwidth}
        \centering
        \includegraphics[width=\linewidth]{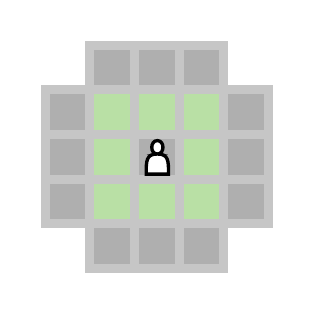}
        \caption{\textsc{Pinko Pallino}.}
    \end{subfigure}
    \hfill
    \begin{subfigure}[t]{0.22\textwidth}
        \centering
        \includegraphics[width=\linewidth]{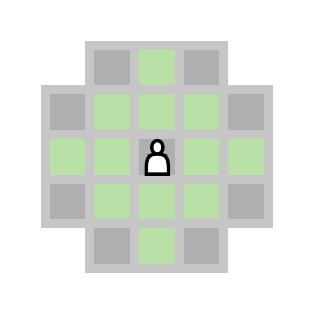}
        \caption{\textsc{Blockade}.}
    \end{subfigure}
    \hfill
    \caption{Allowed pawn movements across the games.}
    \label{fig:moves}
\end{figure}

\begin{table}[h]
\centering
\renewcommand{\arraystretch}{1.25}
\begin{tabular}{@{}lcccc@{}}
\toprule
 & \textsc{Quoridor} & \textsc{Maze Attack} & \textsc{Pinko Pallino} & \textsc{Blockade} \\
\midrule
Pawns per player & $1$ & $1$ & $1$ & $2$ \\
Pawn movement    & $1$ orth. & $1$ orth. & $1$ orth./diag. & $2$ orth.\ or $1$ diag. \\
Turn             & move or place & move or place & move or place & move and place \\
Confrontation    & jump ($+$ deviation) & no jump (block) & ? & jump \\
Wall types       & one & one & one & two (vert., horiz.) \\
\bottomrule
\end{tabular}
\caption{The four games as combinations of independent rule parameters. All of them share the same wall mechanics and the non-blocking constraint; only the listed axes differ.}
\label{tab:games}
\end{table}

In Table \ref{tab:games}, the four games read off compactly. \textsc{Quoridor} is one pawn, move-or-place, a single orthogonal step, the jump-with-deviation rule, and a universal wall type.\\ \textsc{Maze Attack} is identical, except that pawn confrontation is blocking rather than jumping. \textsc{Pinko Pallino} is identical to \textsc{Quoridor}, except that pawns may also step diagonally. \textsc{Blockade} is two pawns, move-and-place, a two-orthogonal-or-one-diagonal step with jumps, and two directional wall types.

Crucially, our reduction is built entirely from features shared by all four games---the race to the opposite side and the length-two walls---which is why the single construction of Section~\ref{sec:reduction} serves them all.

\subsection{Generalization and the decision problem}
\label{sec:generalization}
For complexity purposes each game is generalized from its fixed commercial board to an arbitrary $m\times n$ board, with the number of walls available to each player linearly bounded in $n+m$. The four games share the same notion of position and differ only in a few rule switches, so we define a single decision problem and let the game itself be a parameter, rather than repeating four almost identical encodings.

\paragraph{Race-and-wall position (common to all four games).}
A \emph{race-and-wall position} consists of an $m\times n$ board possibly already containing walls, together with
\begin{itemize}
\item the positions of the pawns, $p$ per player (since \textsc{Blockade} requires two pawns per player);
\item the number of walls each player still has available, possibly split into vertical and horizontal ones (since \textsc{Blockade} has two kinds of walls);
\item the player to move, White or Black.
\end{itemize}
The wall mechanics---length-two walls, no crossing or overlapping, permanence, and the non-blocking constraint---are common to all games and are as described above. A position is \emph{valid} if the placed walls are pairwise non-crossing and every pawn can still reach its target side. A valid position is a \emph{valid configuration of} $\mathsf{Game}\in\{\textsc{Quoridor},\textsc{Maze Attack},\textsc{Pinko Pallino},\textsc{Blockade}\}$ if it has the number of pawns and the wall types prescribed by $\mathsf{Game}$. 

\begin{definition}
For $\mathsf{Game}\in\{\textsc{Quoridor},\textsc{Maze Attack},\textsc{Pinko Pallino},\textsc{Blockade}\}$ we consider
$$
\textsc{Win}(\mathsf{Game})=
\left\{
\mathcal{B} :
\mathcal{B}\text{ is a \textit{valid} configuration of }\mathsf{Game}\text{ and White has a winning strategy}
\right\}.
$$
\end{definition}

\section{Generalized Hex on planar graphs}
\label{sec:pgh}
The source problem for our reduction is not the ordinary $n\times n$ game of Hex, but the generalization of Hex to planar graphs introduced by Reisch on his way to proving that Hex is PSPACE-complete \cite{reisch1981}.\\
\\We recall it here in Reisch's own terms, because two features of his formulation---the planarity of the underlying graph and the degree bound on the uncolored vertices---are exactly what makes our embedding into a \emph{race-and-wall} board possible.

\subsection{Hex-graphs and the game graph-Hex}
Following Reisch, the game is played on a planar graph equipped with two distinguished terminals that lie on the outer face.
\begin{definition}[Hex-graph and outside pair]
Let $H=(V,E)$ be a finite undirected planar graph and let $s,t\in V$ be two distinct vertices. We say that $s$ and $t$ form an \emph{outside pair} if the graph $(V,E\cup\{st\})$ obtained by adding the edge $st$ is still planar; equivalently, $H$ admits a planar embedding in which both $s$ and $t$ are accessible from the outer face. A planar graph together with a selected outside pair $(s,t)$ is called a \emph{Hex-graph}.
\end{definition}
\noindent The terminals $s$ and $t$ are never played on: they merely mark the two endpoints that White must connect.
\begin{definition}[graph-Hex]
\emph{graph-Hex} is played on a Hex-graph $H=(V,E)$ with outside pair $(s,t)$. White and Black alternately place a stone of their colour on a vertex of $V\setminus\{s,t\}$, one per turn; a vertex may be played at most once, and a stone is never moved or removed. White moves first and \emph{wins} if at some point the white vertices contain a path from $s$ to $t$; otherwise Black wins.
\end{definition}
\noindent A \emph{position} (or play situation) of graph-Hex is a Hex-graph in which some vertices already carry a stone. We encode it by a partial colouring $\chi\colon V(H)\setminus\{s,t\}\to\{W,B,\Open\}$, where $\Open$ marks the \emph{open} (still-uncoloured) vertices; the terminals $s,t$ are never coloured and are not counted among the open vertices.

\subsection{The source language}
We reduce from the exact language that Reisch proves PSPACE-complete \cite{reisch1981}. Beyond requiring White to have a winning strategy, it imposes a degree restriction on the open vertices.
\begin{definition}
\GHex is the set of graph-Hex positions $\langle H,s,t,\chi\rangle$ such that
\begin{enumerate}
\item White has a winning strategy; 
\item every open vertex has degree at most $3$.
\end{enumerate}
\end{definition}
\noindent This is exactly Reisch's restriction: since the terminals are not open vertices, condition~(2) bounds the degree of every open (uncoloured) vertex, which is the property our reduction relies on.
\begin{theorem}[Reisch \cite{reisch1981}]\label{thm:reisch}
\GHex $\,$ is PSPACE-complete: deciding whether White has a winning strategy in a graph-Hex position is PSPACE-complete even under the restriction that every open vertex other than the two terminals has degree at most three.
\end{theorem}

\begin{lemma}[Determinacy of graph-Hex]\label{lem:determinacy}
graph-Hex is a finite, two-player, perfect-information game without draws such that every play halts
after at most $|V(H)|$ moves. Consequently, from any
position exactly one of the two players has a winning strategy.
\end{lemma}
\begin{proof}
At most one stone is ever placed on each vertex of $V(H)\setminus\{s,t\}$ and stones are
never removed, so every play halts after at most $|V(H)|$ moves. By the winning condition the
outcome of a halted play is decided---White wins exactly when an $s$--$t$ path of white
vertices exists, and Black wins in every other case---so no play ends in a draw. A finite,
perfect-information game without draws is determined: by backward induction on its finite game tree, we get that White has a winning strategy precisely if and only if Black has none.
\end{proof}

\section{Embedding graph-Hex into the board}
\label{sec:embedding}
We translate a graph-Hex instance into a board configuration using only walls. This is done in two stages. First, we fix a geometric layout of the graph $H$ on the board; then, we replace each vertex and each edge of this layout by a local \emph{gadget} made of walls. Since the construction uses only wall placements, the same translation is valid for all the game variants considered here. Throughout, the figures display the exact wall geometry of each gadget, while the text explains why the gadget behaves as required.

\paragraph{The layout.}
The layout is provided by a standard planar grid embedding.
\begin{fact}[Orthogonal grid embedding \cite{tamassiatollis1989}]\label{fact:ortho}
Every $n$-vertex planar graph of maximum degree at most four admits an \emph{orthogonal grid drawing} of area $O(n^2)$, computable in polynomial time: each vertex is placed on a distinct grid point and each edge is drawn as a path of axis-parallel grid segments, so that distinct edges meet only at shared endpoints and bend only at grid points.
\end{fact}
\begin{lemma}[Degree reduction]\label{lem:degree}
Every graph-Hex position $\langle H,s,t,\chi\rangle$ can be transformed, in time polynomial in $|V(H)|$, into an equivalent graph-Hex position in which no vertex is coloured black, every vertex has degree at most three, and the two terminals lie on the outer face.
\end{lemma}
\begin{proof}
Each of the three transformations below preserves the winner, hence yields an equivalent position.

\emph{Black vertices.} A vertex coloured black is permanently unusable by White, since a black stone is never removed; deleting it together with its incident edges therefore leaves the winner unchanged, and we may assume no vertex is coloured black.

\emph{Terminals.} If a terminal, say $s$, has degree greater than three, add a new vertex $s'$ joined to $s$ by a single edge, recolour $s$ white, and let $s'$ be the new terminal (and likewise for $t$). As $s'$ is adjacent only to the now-white $s$, a white path from $s'$ to $t'$ exists if and only if the original position has a white path from $s$ to $t$, so the winner is unchanged. Since $s$ and $t$ form an outside pair, $s'$ and $t'$ can be placed on the outer face, where they have degree one.

\emph{High-degree vertices.} After these steps every vertex of degree greater than three is coloured white. We replace each such white vertex by a tree of degree-three white junctions; as a connected set of white vertices behaves like a single vertex, this preserves the winner. A vertex of degree $d$ yields $O(d)$ junctions, and since $H$ is planar the degrees sum to $O(|V(H)|)$, so the transformation is polynomial and adds only linearly many vertices.
\end{proof}
By Lemma~\ref{lem:degree} we may assume that $H$ has no black vertex, has maximum degree three, and has its two terminals on the outer face. We fix an orthogonal drawing of $H$ as in Fact~\ref{fact:ortho}, whose bounding rectangle measures $w\times h$ \emph{grid units}, with $w,h=O(|V(H)|)$.

We realize the drawing on the board by replacing each grid edge with an \emph{edge gadget} (a corridor of walls), each grid vertex with a \emph{vertex gadget}, and by bending corridors wherever the drawn edge bends. The rest of this subsection describes the gadgets.

\subsection{Edge/corridor gadgets}
An edge gadget is a corridor one cell wide, flanked on both sides by short \emph{rungs}: length-two wall segments protruding into the cells adjacent to the channel. We distinguish edge gadgets by the position of these rungs. If the rungs on the two sides face each other at the same height, we call that portion of the corridor \emph{aligned}; if they alternate between the two sides, we call it \emph{staggered}. Both types appear because the two ports of a vertex gadget may demand different local configurations.

Since some connections require an aligned corridor and others a staggered one, Figure~\ref{fig:fix} shows a small \emph{edge-switch} gadget that converts one into the other; it is inserted on any corridor whose two endpoints demand opposite types. Moreover, because corridors are tiled by length-two walls, the distance between two ports occasionally has the wrong parity to be tiled exactly; in that single case one inserts a \emph{parity-fix} gadget, which absorbs a one-cell offset (illustrated in the embedding example of Figure~\ref{fig:embedding}). Both the edge-switch and the parity-fix are themselves non-blockable, for the same reason as a plain corridor.
\begin{lemma}[Corridors are non-blockable]\label{lem:corridor}
For each of the configurations shown in Figures~\ref{fig:fix} and~\ref{fig:turn}, there is no legal way to obstruct the gadget and prevent a pawn from passing through it.
\end{lemma}
\begin{proof}
To stop a pawn in a one-cell-wide channel, a player must lay a wall across the channel, i.e.\ on a grid line separating two consecutive channel cells. Such a wall has length two and therefore extends one further unit into a cell adjacent to the channel. By construction the rungs occupy, on every grid line that crosses the channel, exactly that adjacent segment on at least one side (on both sides in the aligned case, on the alternating side in the staggered case, and likewise around a bend and through the switch). Any cross-wall would thus overlap or cross an existing wall, which is illegal. Hence no legal wall blocks the channel, and the pawn can always pass.
\end{proof}
\begin{figure}[ht]
    \centering
    \includegraphics[width=0.45\linewidth]{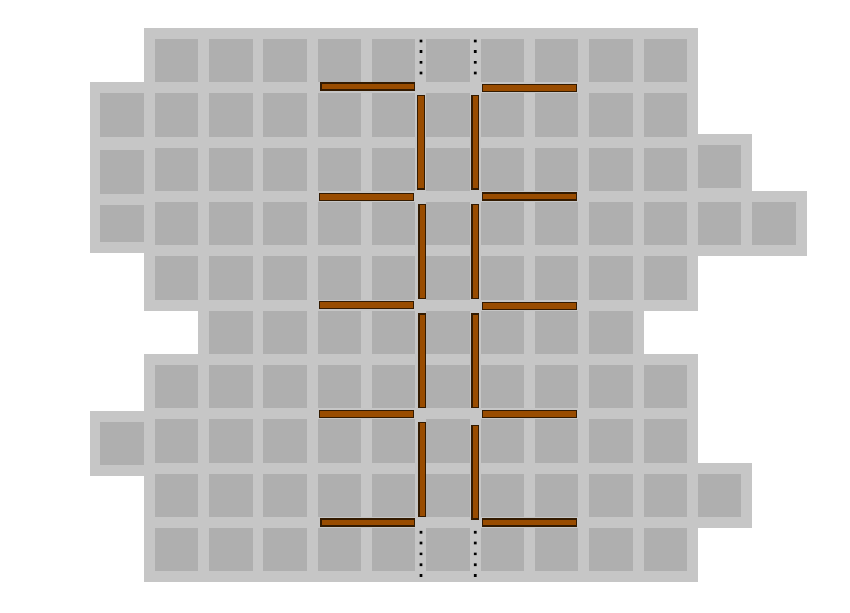}
    \includegraphics[width=0.30\linewidth]{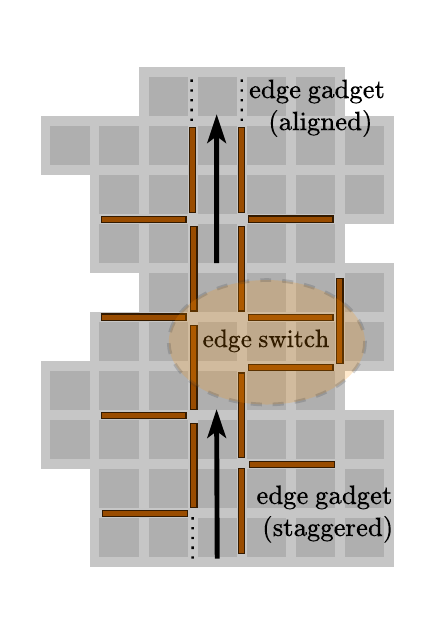}
    \caption{Left: a non-blockable \emph{aligned} corridor that can also be used as an edge gadget. Right: a gadget used to transform a staggered corridor into an aligned corridor.}
    \label{fig:fix}
\end{figure}
\begin{figure}[h]
    \centering
    \includegraphics[width=0.4\linewidth]{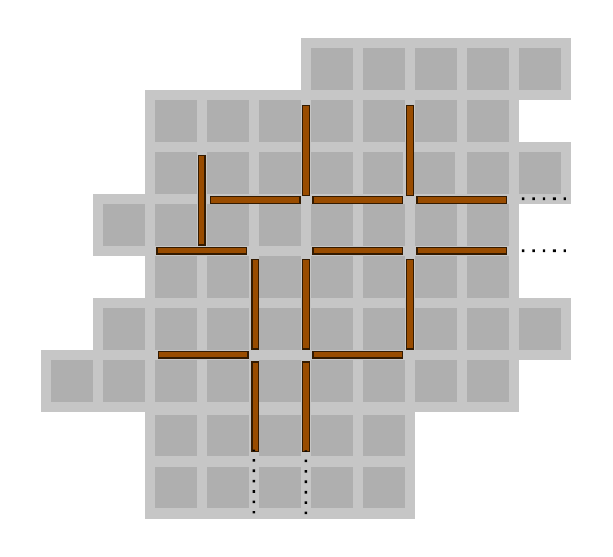}
    \includegraphics[width=0.4\linewidth]{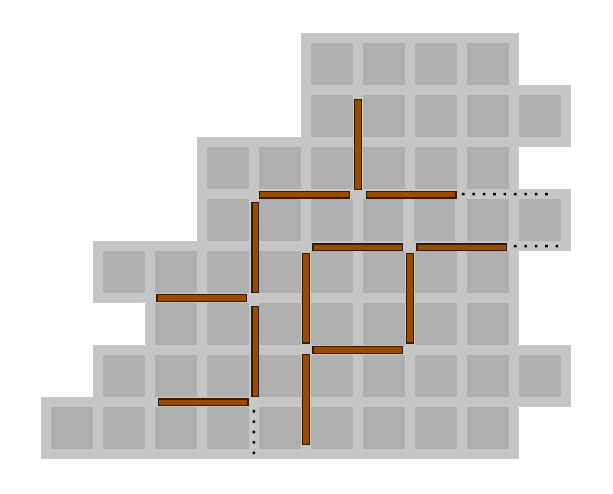}
    \caption{Left: right turn of an aligned edge gadget. Right: right turn of a staggered edge gadget.}
    \label{fig:turn}
\end{figure}
\subsection{Open vertex gadget}
Every still-uncoloured vertex has degree at most three, and its gadget is shown in Figure~\ref{fig:free}. It consists of a small central chamber into which the (at most three) incident corridors open. In the empty gadget the chamber is open enough that the incident corridors are not yet mutually connected in a way that prevents the opponent from severing them: White must spend a wall to secure a passage, and Black two walls to seal it off. This is exactly what lets the gadget mimic the colouring of $v$. Vertices of degree two are built in the same way, with one of the three ports already walled off; a degree-one vertex never needs to be traversed and may be treated as a dead end.

\begin{figure}[h]
    \centering
    \includegraphics[width=0.49\linewidth]{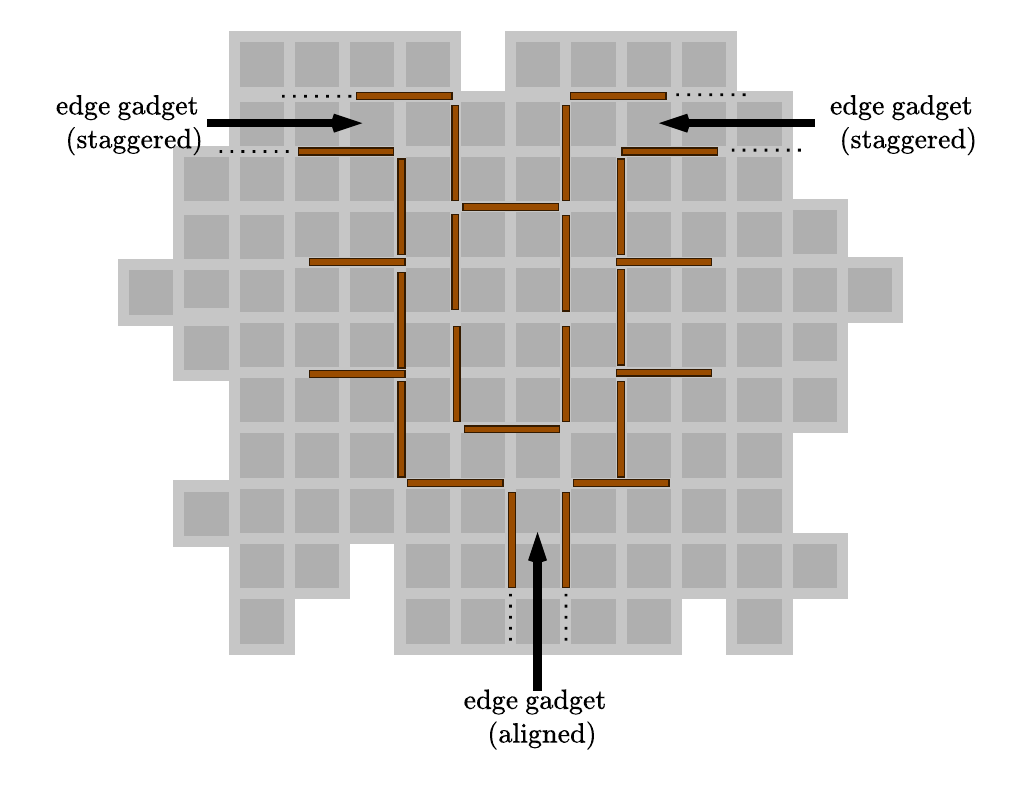}
    \includegraphics[width=0.44\linewidth]{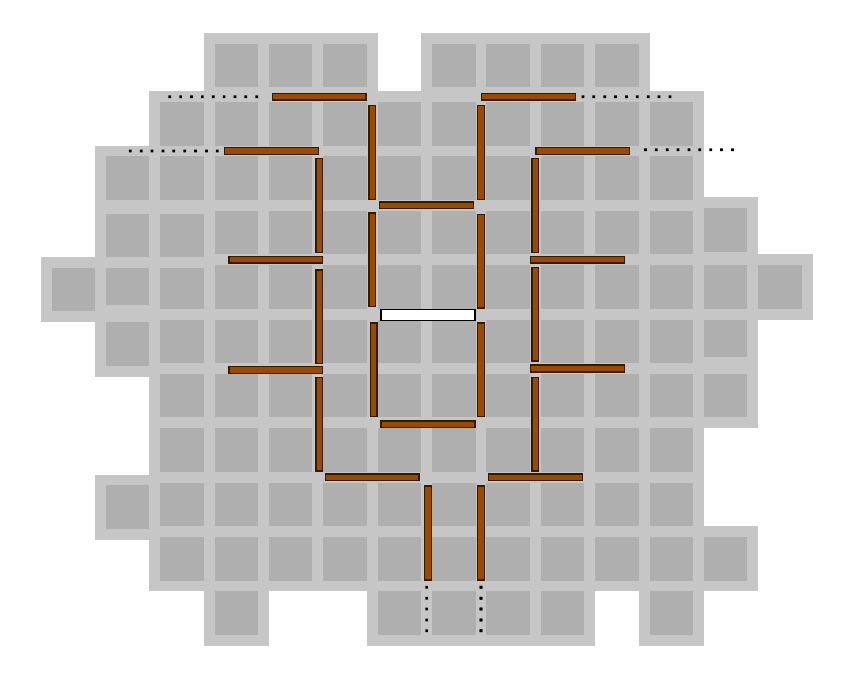}
    \caption{Left: gadget for an unoccupied vertex $v$ with $\deg(v)=3$. Right: vertex gadget colored by White.}
    \label{fig:free}
\end{figure}

\begin{lemma}[White colours a vertex]\label{lem:white}
By placing the single white wall shown in Figure~\ref{fig:free}, White turns the chamber into a passage that cannot afterwards be obstructed by Black. Equivalently, placing this wall \emph{saves} the vertex, and White can use it on a route to $t$.
\end{lemma}

\begin{proof}
After the white wall is placed, the chamber together with its boundary walls forms a channel one cell wide whose flanking grid lines are already occupied by walls, exactly as in a corridor. By the argument of Lemma~\ref{lem:corridor}, no further wall can be laid across it, so the connection through the vertex is permanent.
\end{proof}

\begin{lemma}[Black colours a vertex]\label{lem:black}
If both black walls are placed as in Figure~\ref{fig:Bbreak}, no pair of incident edge gadgets can be connected through the chamber. Moreover, after Black places the first of these two walls, the vertex is already virtually blocked: White cannot prevent Black from later placing the second wall. 
\end{lemma}
\begin{proof}
With both walls in place, the chamber is partitioned so that the openings of distinct corridors lie in regions separated by walls. Hence no path can connect two incident edge gadgets through the chamber; see Figure~\ref{fig:Bbreak}.
It remains to justify the commitment claim. Once the first black wall has been placed, it is irreversible by wall permanence. The second wall remains a legal placement for Black regardless of White's reply (its legality with respect to the non-blocking constraint is ensured by the safe corridor; see Lemma~\ref{lem:legal}). Indeed, any white wall that would pre-empt that position would either cross the first black wall, which is illegal, or close the chamber, which only helps Black. Thus White cannot prevent Black from completing the seal. 

\end{proof}
\begin{figure}[h]
    \centering
    \includegraphics[width=0.45\linewidth]{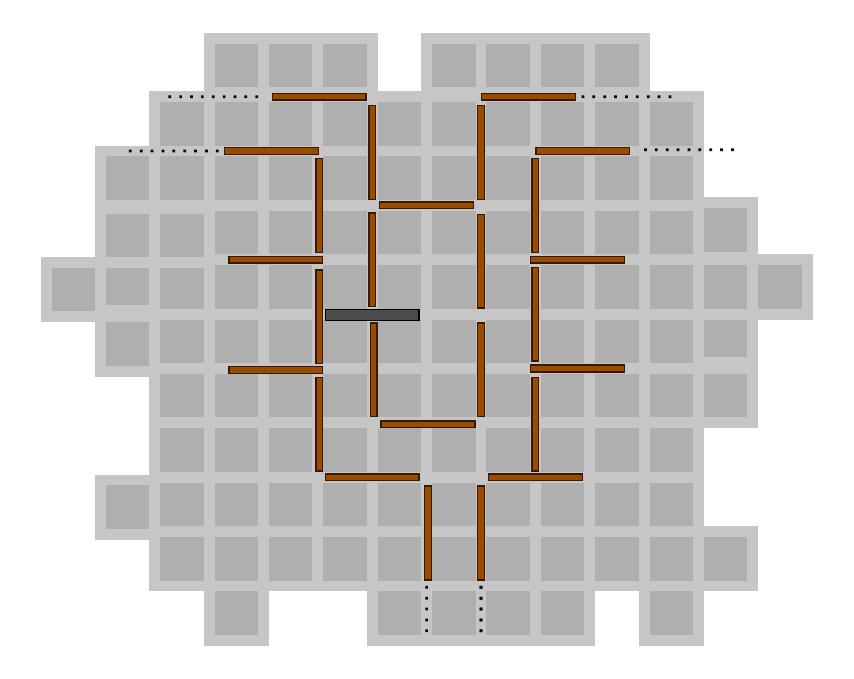}
    \includegraphics[width=0.45\linewidth]{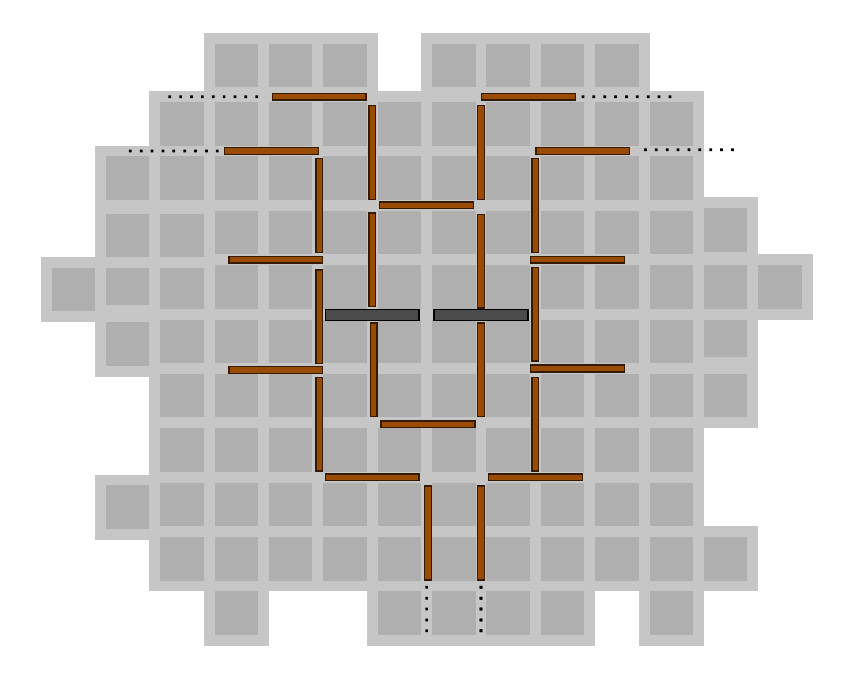}
    \caption{Left: a partially blocked vertex gadget. Right: a vertex gadget colored by Black.}
    \label{fig:Bbreak}
\end{figure}

\subsection{Already coloured vertices and terminals}
By Lemma~\ref{lem:degree} we may assume the position has no black vertex, that every vertex has degree at most three, and that the terminals lie on the outer face; it remains to describe the gadgets for the white vertices and the terminals.

Every vertex coloured white is permanently available to White, and we realize it as a \emph{non-blockable junction}: a small open region of fixed size where four corridors meet and through which every pair of \emph{open} corridors is connected, with no legal wall placement able to disconnect them (Figure~\ref{fig:junction}). The gadget always has the same size; the degree of the vertex only determines how many of the four corridors are actually open, the remaining ones being walled off. Its non-blockability follows from Lemma~\ref{lem:corridor}. The same gadget realizes the degree-three junctions produced by the decomposition of Lemma~\ref{lem:degree}.
\begin{figure}[h]
    \centering
    \includegraphics[width=0.35\linewidth]{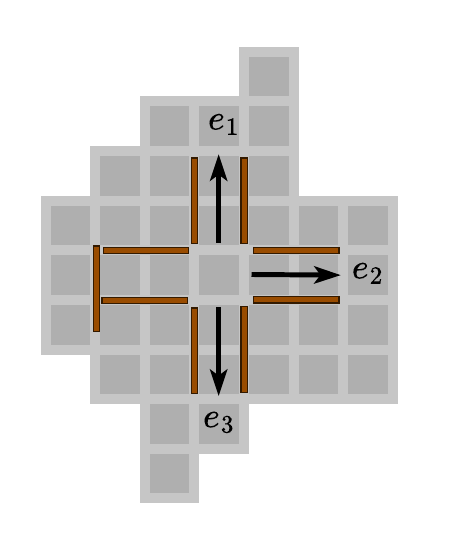}
    \includegraphics[width=0.38\linewidth]{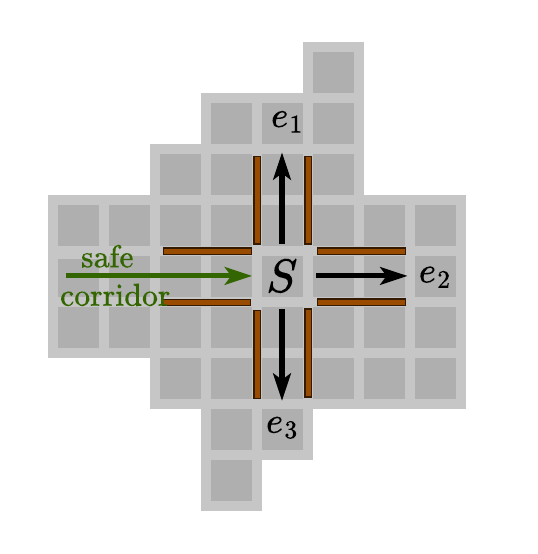}
    \caption{The junction gadget. Left: a degree-three junction for a white vertex.\\ Right: a degree-three terminal, whose fourth corridor is the safe corridor.}
    \label{fig:junction}
\end{figure}

A terminal is realized by the same junction gadget: its incident corridors occupy that many of the four positions, and the safe corridor of Section~\ref{sec:reduction} uses one of the remaining ones (Figure~\ref{fig:junction}, right).
\noindent Having described how each local gadget is translated, we can now represent a whole graph-Hex instance inside the board, as in Figure~\ref{fig:embedding}.
\begin{figure}[h]
    \centering
    \begin{minipage}[c]{0.32\linewidth}
        \centering
        \resizebox{\linewidth}{!}{%
        \begin{tikzpicture}[x=1.4cm,y=1.4cm,
          vopen/.style ={draw,thick,rectangle,fill=gray!55,minimum size=8mm,inner sep=0pt},
          vwhite/.style={draw,thick,circle,fill=white,minimum size=8mm,inner sep=0pt},
          vblack/.style={draw,thick,circle,fill=black,text=white,minimum size=8mm,inner sep=0pt},
          gline/.style ={gray!35,very thin},
          gedge/.style ={thick,line cap=round}]
          \foreach \x in {-1,...,5} \draw[gline] (\x,-1) -- (\x,4);
          \foreach \y in {-1,...,4} \draw[gline] (-1,\y) -- (5,\y);
          \node[vwhite] (s) at (0,2) {$s$};
          \node[vopen]  (q) at (2,3) {};
          \node[vwhite] (t) at (4,2) {$t$};
          \node[vwhite] (m) at (2,1) {};
          \node[vblack] (b) at (2,0) {};
          \draw[gedge] (s) -- (0,3) -- (q);
          \draw[gedge] (t) -- (4,3) -- (q);
          \draw[gedge] (q) -- (m);
          \draw[gedge] (m) -- (b);
          \draw[gedge] (s) -- (0,1) -- (m);
          \draw[gedge] (t) -- (4,1) -- (m);
        \end{tikzpicture}}
    \end{minipage}\hfill
    \begin{minipage}[c]{0.64\linewidth}
        \centering
        \includegraphics[width=\linewidth]{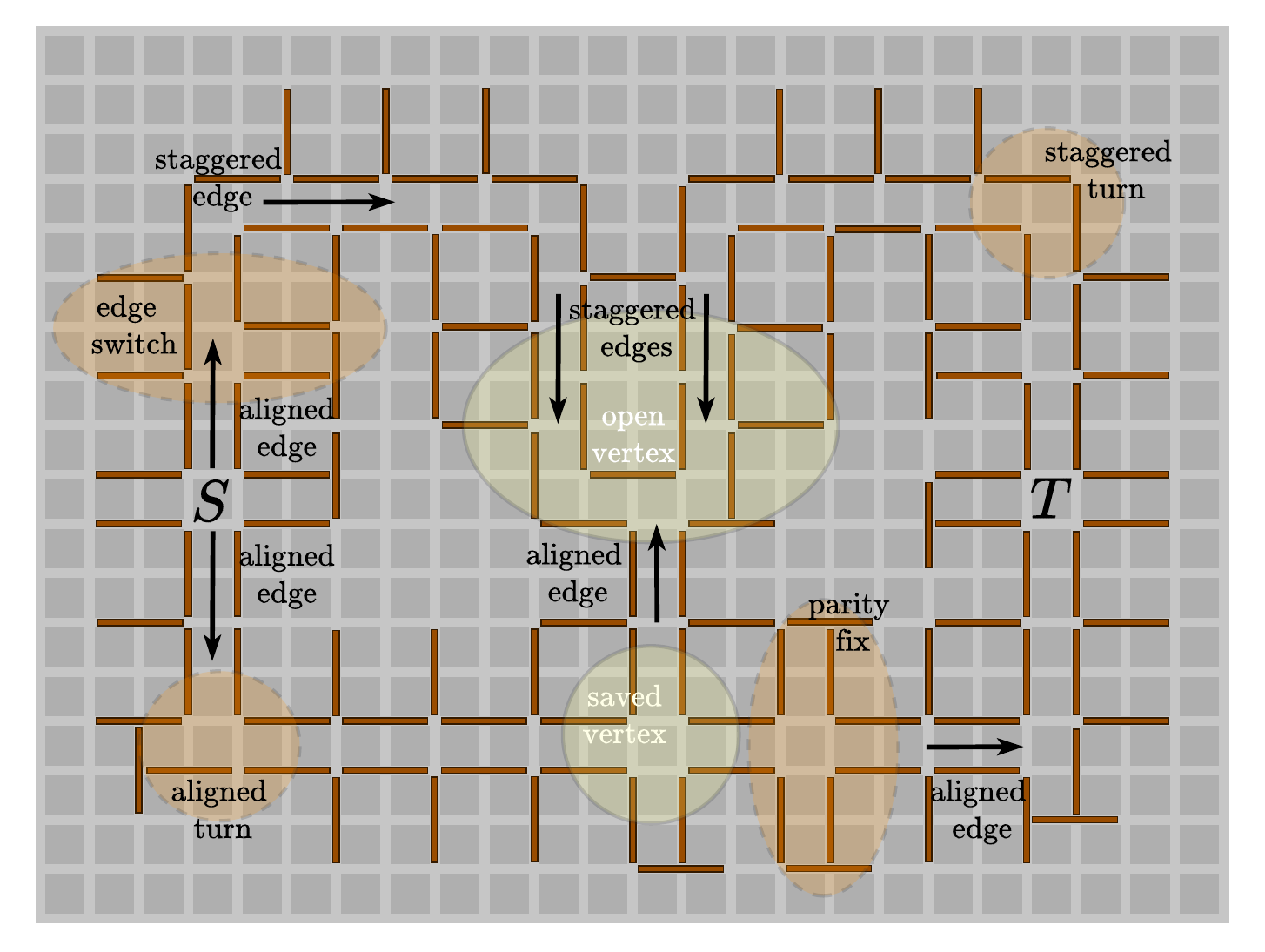}
    \end{minipage}
    \caption{On the left, a graph-Hex instance: gray squares are open (uncoloured) vertices, circles are coloured vertices; white circles are coloured by White (the terminals $s,t$ among them) and black circles are coloured by Black. On the right, its board embedding.}
    \label{fig:embedding}
\end{figure}
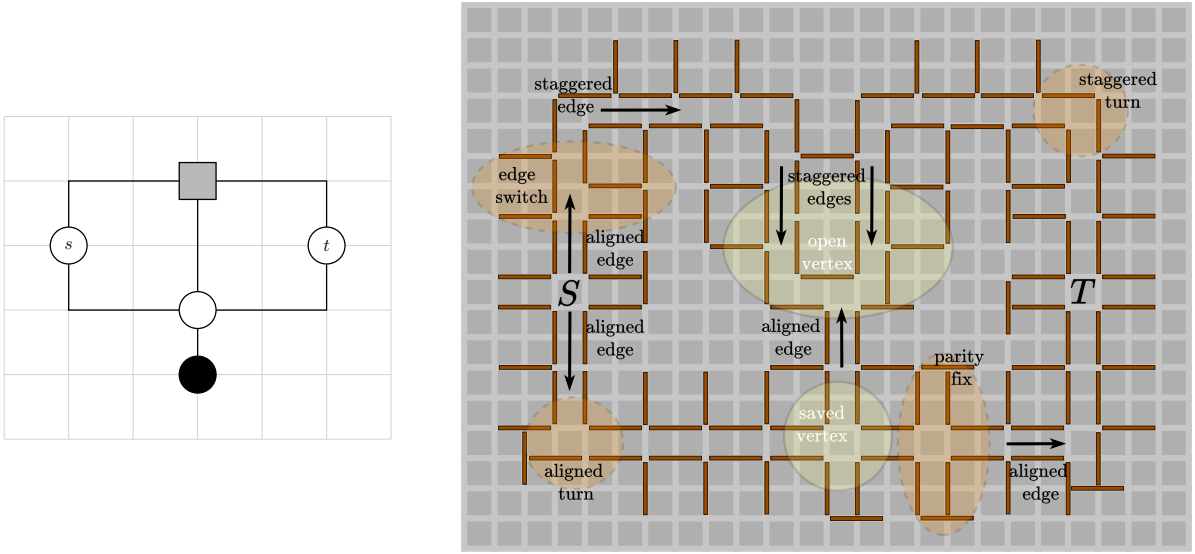

\begin{lemma}[Polynomial-time embedding]\label{lem:embedding}
Given a graph-Hex position $\langle H,s,t,\chi\rangle$, the board embedding described above can be constructed in time polynomial in $|V(H)|$. It occupies a board of dimensions $\gwidth(H)\times\gheight(H)$ with $\gwidth(H)\le 11w$ and $\gheight(H)\le 11h$, where $w,h=O(|V(H)|)$ are the side lengths of the orthogonal grid drawing of $H$.
\end{lemma}
\begin{proof}
By Fact~\ref{fact:ortho}, the orthogonal drawing of $H$ (after the degree reduction of Lemma~\ref{lem:degree}) lies in a $w\times h$ rectangle with $w,h=O(|V(H)|)$ and is computable in polynomial time. Each gadget has constant size (at most $11\times 11$ cells), so replacing every vertex and edge of the drawing by its gadget enlarges each grid unit by a constant factor; hence the board has dimensions $\gwidth(H)\le 11w$ and $\gheight(H)\le 11h$, both polynomial in $|V(H)|$, and writing down the resulting walls takes time polynomial in the board size.
\end{proof}

\section{The reduction}\label{sec:reduction}
It remains to turn the embedding of Lemma~\ref{lem:embedding} into a full configuration by choosing the safe corridors, the pawn positions and the wall budgets. We describe a single configuration that serves \textsc{Quoridor}, \textsc{Maze Attack}, and \textsc{Pinko Pallino} at once: it is built so as to remain valid whether or not diagonal moves are allowed. The configuration for \textsc{Blockade} is then obtained from this one by adding only the parts needed for the second pawn of each player. We now carry out the construction for the first three games.

Given a graph-Hex position on a graph $H$, the corresponding board position is obtained from the embedding of $H$ of Lemma~\ref{lem:embedding}. The embedded graph occupies a rectangle of size $\gwidth(H)\times\gheight(H)$, positioned so that the terminal $t$ lies at distance $2$ from White's target side and is joined to it by a single pair of walls forming an unbreakable corridor. The two pawns are then placed in two distinct connected components of the board.

\paragraph{White component.}
We orient the board so that moving \emph{left} goes toward Black's target side and moving \emph{right} toward White's target side, while \emph{up} and \emph{down} denote the two perpendicular directions. Starting from $s$, we construct an unblockable safe corridor. The corridor first goes straight to the left (toward Black's target side) for $2\cdot |V(H)|$ cells; call the endpoint of this segment $Q$. From $Q$, the corridor goes downward and then continues to the right for another $2\cdot |V(H)|$ cells, reaching a point $A$ directly below $s$.
From $A$, the corridor continues as a snake-like path placed exactly below the embedding: it moves to the right for $\gwidth(H)$ cells, then goes down and moves to the left for another $\gwidth(H)$ cells, then goes down again and moves to the right, and so on. This back-and-forth pattern is repeated $2\cdot \gheight(H)$ times, each horizontal segment having length $\gwidth(H)$ and lying directly below the embedding. Two consecutive horizontal corridors lie at distance $2$ from each other and are therefore joined by a vertical segment of length exactly $4$ (the value used for the vertical segments in Remark~\ref{rem}).
After the last such segment, the corridor continues toward Black's target side for another $2\cdot |V(H)|$ cells, reaching point $B$. It then turns downward and back toward White's target side for at least another $2\cdot |V(H)|+\gwidth(H)$ cells, remaining unblockable throughout.
\paragraph{Black component.}
The black component is obtained by copying the snake-like path from $A$ to $B$ and placing this copy on the left, at distance $2\cdot |V(H)|$ from Black's target side. Let $A'$ and $B'$ denote the points corresponding to $A$ and $B$ in this copied path.
Unlike in White's component, the corridor from $B'$ continues directly toward Black's target side for another $2\cdot |V(H)|$ cells, again remaining unblockable throughout.
\paragraph{Pawns, walls, and turn order.}
The white pawn is placed at $Q$ and the black pawn at $A'$, and White has $|V(H)|$ walls to place while Black has $2\cdot |V(H)|$ walls. The player to move in the board position is the same as in the graph-Hex position. The number and placement of the pawns and the division of the wall budget are the only game-dependent parts of the construction; they are adjusted for \textsc{Blockade} below.\\ Because the two pawns lie in separate components, we have full control over the distance between them: increasing it as a function of $\gwidth(H)$ and $\gheight(H)$ enlarges the board by any polynomial amount without touching the gadgets or the logic of the reduction. Hence in variants where the number of walls depends on the board size, the budget can be scaled to any polynomial in $\gwidth(H)$ and $\gheight(H)$, the linear case included.

\begin{remark}\label{rem}
    
Let $d_1$, $d_2$, and $d_3$ be the number of moves required for a pawn to reach $B$ from $A$ under, respectively, orthogonal-only movement (\textsc{Quoridor}, \textsc{Maze Attack}), single-cell movement with diagonals (\textsc{Pinko Pallino}), and \textsc{Blockade} movement (two orthogonal steps or one diagonal step). The path from $A$ to $B$ consists of the snake---$2\gheight(H)$ horizontal segments of length $\gwidth(H)$ joined by $2\gheight(H)-1$ vertical connectors---followed by a straight segment of $2\cdot|V(H)|$ cells reaching $B$. A connector costs $4$ orthogonal moves but only $2$ once diagonals are allowed (it spans a net vertical distance of $2$); the straight parts cost the same either way; and a \textsc{Blockade} move covers at most two of the single-cell steps counted by $d_2$. Hence
$$
\begin{aligned}
d_1 &= (\gwidth(H)+4)\bigl(2\gheight(H)-1\bigr) + \gwidth(H) + 2\cdot|V(H)|,\\
d_2 &= (\gwidth(H)+2)\bigl(2\gheight(H)-1\bigr) + \gwidth(H) + 2\cdot|V(H)|,\\
d_3 &\geq \tfrac12\, d_1 \geq \gwidth(H)\times\gheight(H) + |V(H)|,
\end{aligned}
$$
and in particular $d_1 \geq d_2 \geq 2 \cdot\gwidth(H)\times\gheight(H) + 2\cdot|V(H)|$.
\end{remark}

\begin{lemma}[Wall placements are legal]\label{lem:legal}
In the constructed position each pawn always retains the safe corridor of its own component as a route to its target side. Consequently no placement of a wall inside a vertex or edge gadget ever removes a pawn's last route, so every such placement satisfies the non-blocking constraint and is legal.
\end{lemma}
\begin{proof}
The safe corridors of the two components are unblockable (Lemma~\ref{lem:corridor}) and are disjoint from every gadget wall; hence, whatever walls are placed inside the gadgets, each pawn keeps a route to its target side along its safe corridor. The non-blocking constraint forbids only a wall that would remove a pawn's last route, so it never forbids a gadget wall.
\end{proof}

\begin{figure}[h]
\centering
\includegraphics[width=1.0\linewidth]{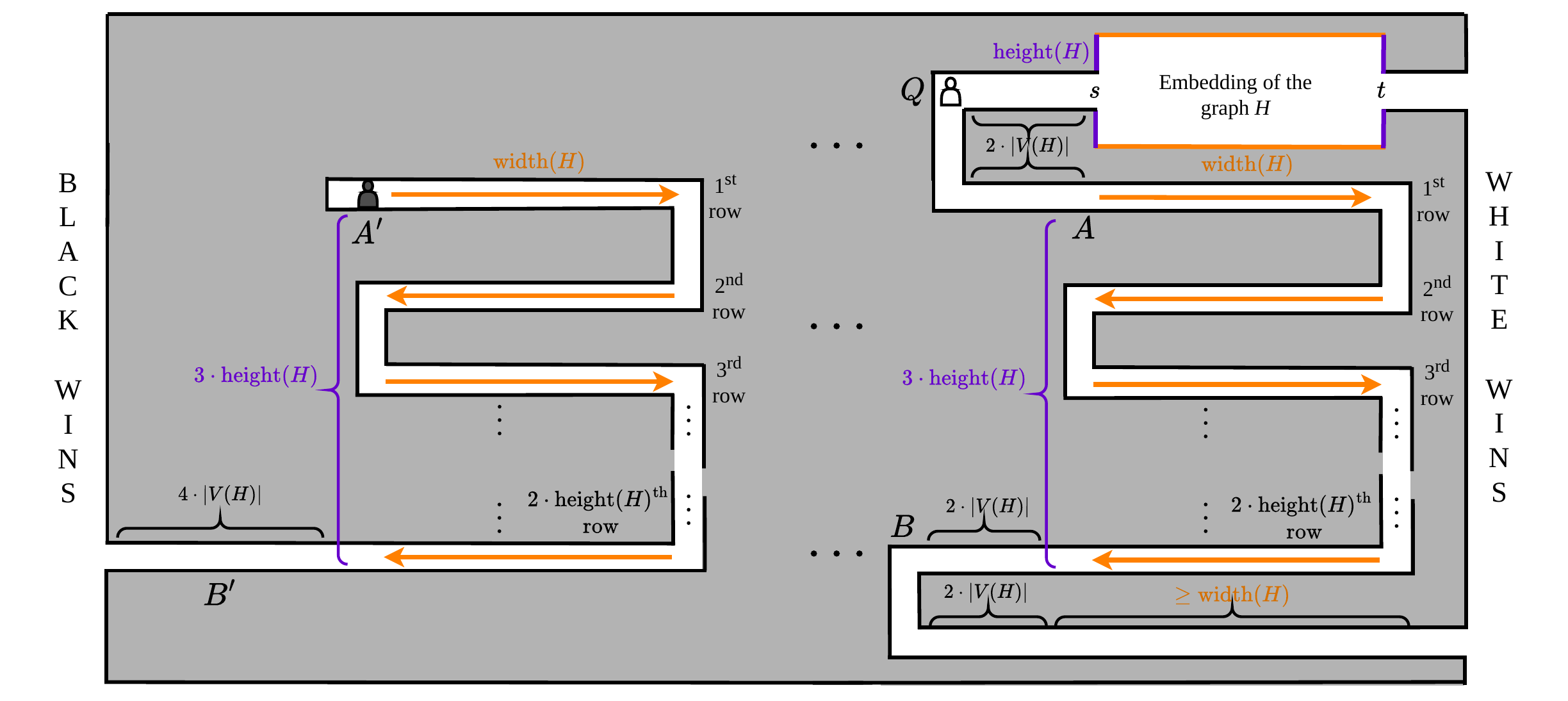}
\caption{Visual summary of the reduction.}
\label{fig:reduction}
\end{figure}
At a high level: if White wins the embedded graph-Hex game, White secures a short $s$--$t$ path---of length less than $\gwidth(H)\times\gheight(H)$, since the whole embedding fits in that many cells---ending near White's target side, and reaches it before Black can traverse its corridor. If White loses, White's only route is the safe corridor; but Black's copy of it is shorter (it omits the tail that brings White's corridor back to White's target side) and the black pawn starts ahead, so Black arrives first. Theorem~\ref{thm:main} below makes this precise.
\begin{lemma}[The confrontation rule is irrelevant]\label{lem:nojump}
In any play of the constructed position, no confrontation between the two pawns can ever occur. Consequently, the outcome does not depend on the confrontation rule.
\end{lemma}

\begin{proof}
By construction, the two pawns are placed in two distinct connected components of the board. Therefore, neither pawn can ever reach the component containing the other pawn, and in particular the two pawns can never interact directly. Hence no confrontation rule can ever be triggered, and the outcome of the game is independent of such rules.
\end{proof}

\begin{theorem}\label{thm:main}
Given a \GHex position, consider the board position produced by the reduction. White has a winning strategy in graph-Hex if and only if White has a winning strategy in \textsc{Quoridor}; by Lemma~\ref{lem:nojump} the same equivalence holds for \textsc{Maze Attack}.
\end{theorem}
\begin{proof}
Assume first that White has a winning strategy in graph-Hex. Then White can follow the corresponding strategy inside the \textsc{Quoridor} embedding by placing one wall for each vertex that the graph-Hex strategy colors white. Thus, regardless of Black's strategy, White can secure an $s$--$t$ path that cannot later be destroyed within at most $|V(H)|$ moves. If Black follows the graph-Hex simulation, this is exactly White's winning strategy in graph-Hex. If Black deviates from the simulation---in particular if Black races along its corridor instead of colouring a vertex---then Black effectively skips moves in the simulated graph-Hex game; since an extra tempo never hurts the player who receives it in graph-Hex, White's winning strategy remains winning, and the path is secured no later.
During these turns, Black may instead move along its corridor. The black pawn starts at $A'$, from which reaching its target side requires at least $d_2 + 2\cdot|V(H)|$ moves (the path from $A'$ to $B'$, followed by the finale straight extension of length $2\cdot|V(H)|$), regardless of whether diagonal movements are allowed. Counting White's own turns instead, securing the path costs at most $|V(H)|$ wall placements, while moving the white pawn from its start $Q$ through $s$ to $t$ and on to White's target side costs at most $2|V(H)|+\gwidth(H)\times\gheight(H)+2$ moves, since a simple $s$--$t$ path inside the embedding has length less than $\gwidth(H)\times\gheight(H)$. White therefore reaches its target side in at most
$$
|V(H)|+\bigl(2\cdot|V(H)|+\gwidth(H)\times\gheight(H)+2\bigr)=\gwidth(H)\times\gheight(H)+3\,|V(H)|+2
$$
of its own turns, whereas by Remark~\ref{rem} the black pawn needs at least $d_2+2\,|V(H)|\ge 2\,\gwidth(H)\times\gheight(H)+4\,|V(H)|$ turns. Since $\gwidth(H)\times\gheight(H)\ge|V(H)|$ (the $|V(H)|$ vertices occupy distinct cells of the embedding),
$$
\gwidth(H)\times\gheight(H)+3\,|V(H)|+2\;<\;2\cdot \bigl(\gwidth(H)\times\gheight(H)\bigr)+4\,|V(H)|,
$$
so White reaches its target side strictly first.
Moreover, Black cannot catch or confront White, cannot obstruct the secured $s$--$t$ path, and cannot reach Black's target side before White reaches White's target side. Therefore, once the simulated winning path has been secured, White follows it and wins the \textsc{Quoridor} game.

\medskip
Suppose for contradiction that White has a winning strategy in the \textsc{Quoridor} instance, but not in the graph-Hex instance. Since graph-Hex admits no draws and is monotone, Black has a winning strategy in graph-Hex. Consider the following \textsc{Quoridor} strategy for Black.
First, Black simulates its winning graph-Hex strategy by placing one wall in each vertex gadget corresponding to a vertex that Black colors in graph-Hex, thereby partially blocking that vertex. This phase lasts for at most $|V(H)|$ turns. Then Black spends at most another $|V(H)|$ turns placing the second wall in each partially blocked vertex gadget, so that every vertex colored by Black in the simulation becomes completely closed. After this, Black moves along the safe corridor toward Black's target side.
We claim that White's winning strategy cannot rely on the safe corridor. Indeed, even if White spent all of the first $2\cdot |V(H)|$ turns moving along that corridor while Black completes the simulated blocking phase, Black would still be closer to its own target side. Hence the safe corridor cannot provide White with a winning route.
Against the strategy described above, after at most $2\cdot |V(H)|$ turns Black has completely closed all vertex gadgets corresponding to vertices that Black can force in the graph-Hex simulation. Since White has no winning strategy in graph-Hex, this eliminates every possible $s$--$t$ path through the embedded graph.
During these same $2\cdot |V(H)|$ turns, White can at most reach the entrance of the embedded graph near $s$, even if diagonal movements are allowed. Once White reaches that point, however, no $s$--$t$ path remains available inside the embedding. White also cannot exploit vertices that were only partially blocked, because Black's second phase closes them before White reaches $s$. Moreover, once a vertex has been partially blocked, White cannot prevent Black from placing the second wall and closing it completely, even if White still has walls available.
Therefore, if White loses the simulated graph-Hex game, the only surviving route is the safe corridor, and that route is losing for White by construction. This contradicts the assumption that White has a winning strategy in \textsc{Quoridor}. Hence any winning strategy for White in the \textsc{Quoridor} instance must pass through a secured $s$--$t$ path in the embedded graph, which yields a winning strategy for White in graph-Hex.
\end{proof}

\begin{corollary}
$\QWIN$ and $\textsc{Win(Maze Attack)}$ are PSPACE-hard.
\end{corollary}
\begin{proof}
The reduction maps every \GHex position $\langle H,s,t,\chi\rangle$ to a valid configuration in polynomial time (Lemma~\ref{lem:embedding}), and by Theorem~\ref{thm:main} it preserves the winner for both \textsc{Quoridor} and \textsc{Maze Attack}.
\end{proof}
\begin{corollary}
\textsc{Win(Pinko Pallino)} is PSPACE-hard.
\end{corollary}
\begin{proof}
\textsc{Pinko Pallino} differs from \textsc{Quoridor} only in admitting diagonal moves. The reduction is unchanged: the proof of Theorem~\ref{thm:main} never relies on jumps, and by Remark~\ref{rem} (the bound on $d_2$) diagonal movement gives neither player a faster route. Hence \GHex $\,$reduces to \textsc{Win(Pinko Pallino)} as well, which is therefore PSPACE-hard.
\end{proof}
\begin{corollary}
\textsc{Win(Blockade)} is PSPACE-hard.
\end{corollary}
\begin{proof}
We adapt the construction of Theorem~\ref{thm:main}. 

Here each player controls two pawns. We place each player's \emph{second} pawn in a private copy of that player's labyrinth; for White this copy omits the graph embedding, so the second white pawn has only the long snake corridor available and can never reach White's target side before the pawn that traverses the embedding, while the second black pawn merely runs a copy of Black's snake. Consequently only one pawn per player is relevant, and the analysis reduces to the one-pawn case. We use the same corridor construction as for \QWIN $\,$ but double the length of the relevant corridors; in particular, the corridor leading White to the entrance $s$ now has length $4\cdot|V(H)|$. Since a \textsc{Blockade} turn moves a pawn and, when possible, places a wall, and a move covers at most two cells (the bound on $d_3$ in Remark~\ref{rem}), these $2\cdot|V(H)|$ turns suffice for the players both to advance along the corridor and to carry out the simulated colouring before the white pawn reaches $s$. Finally, since \textsc{Blockade} distinguishes vertical and horizontal walls, each player is given only $2\cdot|V(H)|$ walls of horizontal type, which are the ones used to colour vertices. With these adjustments the proof proceeds as in the previous cases.
\end{proof}

\begin{lemma}[Membership]\label{lem:pspace}
For each of the four games $\mathsf{Game}$, the language $\textsc{Win}(\mathsf{Game})$ is in PSPACE.
\end{lemma}
\begin{proof}
A position is describable in polynomial space and its legal moves are computable in polynomial time, so it suffices to bound the length of the plays we must examine. Walls are a finite, monotone resource: once placed, a wall is never removed, so along any play the number of wall placements is at most the total wall budget, which is linear in $n$.

We claim that if White has a winning strategy, then White has one in which no configuration ever repeats. Suppose not, and let a winning strategy of White revisit some configuration with White to move. From the second occurrence onward Black could repeat the replies it used after the first occurrence, returning the play to that configuration again and again; the play would then cycle forever and White would never reach its target side, contradicting that the strategy wins. Hence we may restrict attention to White strategies that repeat no configuration.

Configurations sharing a fixed set of walls are only polynomially many---they are determined by the pawn placements and the player to move---and the wall set changes at most linearly many times along a play, so a non-repeating play has polynomial length. Consequently, whether White can force a win can be decided by evaluating a game tree of plays of polynomial length whose positions and legal moves are polynomial-time computable, which lies in PSPACE \cite{hearndemaine2009}. The argument uses only the wall mechanics shared by the four games, so it applies to each of them.
\end{proof}

\begin{theorem}\label{thm:complete}
For each of the four games $\mathsf{Game}$, $\textsc{Win}(\mathsf{Game})$ is PSPACE-complete.
\end{theorem}
\begin{proof}
PSPACE-hardness is established by the corollaries above---Theorem~\ref{thm:main} for \textsc{Quoridor} and \textsc{Maze Attack}, and the corresponding corollaries for \textsc{Pinko Pallino} and \textsc{Blockade}---and membership in PSPACE by Lemma~\ref{lem:pspace}.
\end{proof}

\begin{remark}[The \emph{race-and-wall} family]\label{rem:family}
The reduction settles much more than the four games of Table~\ref{tab:games}. Its only
working ingredients are the race to the opposite side and costant-length walls that cannot cross or
overlap and, once placed, never move: corridor non-blockability is purely geometric
(Lemma~\ref{lem:corridor}), the two pawns sit in separate components and never meet
(Lemma~\ref{lem:nojump}), permanence yields both the irreversible commitment (Lemma~\ref{lem:black})
and the polynomial bound on play length (Lemma~\ref{lem:pspace}), and the non-blocking constraint does not affect the validity of any move in the constructed instances. Essentially, every \textit{race-and-wall} game is PSPACE-complete: whatever the pawn-movement rule (orthogonal or diagonal, any
fixed number of cells per move), whatever the confrontation rule, whatever the number of pawns, with universal or directional walls, and whether or not a player may seal a pawn off from its target side completely.\\The four commercial titles are merely the instances that
happen to have a name.
\end{remark}

\section*{Acknowledgments}
We thank Mirko Giacchini and Alessandro Panconesi for their valuable advice. We are
especially grateful to Benjamin Rin for discussions at FUN~2026 that greatly improved this paper.

\bibliographystyle{alpha}
\bibliography{refs}

\end{document}